\documentclass{lmcs}
\pdfoutput=1
\usepackage[utf8]{inputenc}

\usepackage{lastpage}
\lmcsdoi{20}{1}{8}
\lmcsheading{}{\pageref{LastPage}}{}{}%
{Feb.~10,~2023}{Jan.~26,~2024}{}

\usepackage[colorlinks]{hyperref}
\usepackage{graphicx}
\usepackage{amsmath,amssymb}
\usepackage{mathrsfs}           
\usepackage{mathtools}          
\usepackage{stmaryrd}           
\usepackage{mathpartir}
\usepackage{todonotes}
\usepackage{relsize}
\usepackage{chngcntr}
\usepackage{apptools}
\usepackage{amsthm}
\usepackage{cleveref}
\usepackage{forest}
\usepackage{tikz}
\usetikzlibrary{arrows}
\crefname{equation}{}{}
\crefname{thm}{theorem}{theorems}
\crefname{prop}{proposition}{propositions}
\crefname{lem}{lemma}{lemmas}
\crefname{exa}{example}{examples}
\theoremstyle{definition}
\newtheorem{runexa}{Running example part}
\usepackage{macros}

\begin{document}

\title[Deciding Equations in the Time Warp Algebra]{Deciding Equations in the Time Warp Algebra}
\titlecomment{A precursor to this paper, reporting preliminary results, appeared in the proceedings of RAMiCS 2021~\cite{GGMS21}.}

  \author[S. van Gool]{Sam van Gool\lmcsorcid{0000-0002-6360-6363}}[a]
  \address{IRIF, Universit{\'e} Paris Cit{\'e}, France}
  \email{vangool@irif.fr}
  
  \author[A. Guatto]{Adrien Guatto\lmcsorcid{0000-0001-7961-5611}}[b]
  \address{IRIF, Universit{\'e} Paris Cit{\'e}, France}
  \email{guatto@irif.fr}
 
\author[G. Metcalfe]{George Metcalfe\lmcsorcid{0000-0001-7610-404X}}[c]	
\address{Mathematical Institute, University of Bern,  Switzerland}	
\email{george.metcalfe@unibe.ch}  

\author[S. Santschi]{Simon Santschi\lmcsorcid{0009-0003-9364-5149}}[d]
\address{Mathematical Institute, University of Bern,  Switzerland}	
\email{simon.santschi@unibe.ch}  
 
\thanks{This project has received funding from the European Union’s Horizon 2020 research and innovation programme under the Marie Sk{\l}odowska-Curie grant agreement No 101007627, and Swiss National Science Foundation grant 200021\textunderscore 215157.}	

\keywords{Residuated lattices, Universal algebra, Decision procedures, Graded modalities, Type systems, Programming languages}
\subjclass{F.4.1, I.2.3}


\begin{abstract}
Join-preserving maps on the discrete time scale $\om^+$, referred to as time warps, have been proposed as graded modalities that can be used to quantify the growth of information in the course of program execution. The set of time warps forms a simple distributive involutive residuated lattice---called the time warp algebra---that is equipped with residual operations relevant to potential applications. In this paper, we show that although the time warp algebra generates a variety that lacks the finite model property, it nevertheless has a decidable equational theory. We also describe an implementation of a procedure for deciding equations in this algebra, written in the OCaml programming language, that makes use of the Z3 theorem prover.
\end{abstract}


\maketitle


\section{Introduction}

{\em Graded modalities}~\cite{FujiiKatsumataMellies-2016,GaboardiKatsumataOrchardBreuvartUustalu-2016} provide a unified setting for describing effects of a program, such as which parts of memory it modifies~\cite{GiffordLucassen-1988}, or the resources it consumes, such as how long it takes to run~\cite{GhicaSmith-2014}. Given a {\em type} $A$ and \emph{grading}~$f$, the new type~$\Box_f A$ represents a modification of~$A$ that incorporates the behavior prescribed by~$f$. The gradings themselves are often equipped with an ordered algebraic structure that is relevant for programming applications. Typically, they form a monoid, relating $\Box_{gf} A$ to~$\Box_f \Box_g A$, or admit a precision order along which the graded modality acts contravariantly, i.e., for $f\le g$, there exists a generic program of type~$\Box_g A \to \Box_f A$ that allows movement from more to less precise types. In this case, the order may yield further operations on types; e.g., if the infimum $f\mt g$ of~$f$ and~$g$ exists, it allows values of types~$\Box_f A$ and~$\Box_g B$ to be converted into a value of type~$\Box_{f \mt g} (A \times B)$.

In this paper, we study an algebraic structure induced by the class of graded modalities known as~\emph{time warps}: join-preserving maps on the discrete time scale $\som = \om \cup \{ \om \}$~\cite{Guatto-2018}. Informally, such maps describe the growth of data during program execution, whereby a type~$A$ describes a sequence of sets $A_n$ of values classified by~$A$ at execution step~$n\in\om$, and the type~$\Box_f A$ classifies the set of values of~$A_{f(n)}$ at this step. This approach generalizes a long line of work on programming languages for embedded systems (see,~e.g.,~\cite{CaspiPouzet-1996}) and type theories with modal recursion operators (see,~e.g.,~\cite{Nakano-2000,BirkedalMogelbergSchwinghammerStovring-2012}). 

Let us denote the set of time warps by $W$. Then~$\langle W,\circ,\id\rangle$ is a monoid, where $f\circ g$ (often shortened to $fg$) is the composition of $f,g\in W$, and $\id$ is the identity map. Equipping $W$ with the pointwise order defined by $f\le g\defiff f(n)\le g(n)$, for all $n\in\som$, yields a complete distributive lattice $\langle W,\mt,\jn\rangle$ satisfying for all $f,g_1,g_2,h\in W$, 
\[
f(g_1\jn g_2)h = fg_1h\jn fg_2h\enspace\text{ and }\enspace f(g_1\mt g_2)h = fg_1h\mt fg_2h.
\]
Equivalently, the operation $\circ$ is a {\em double quasi-operator} on $\langle W,\mt,\jn\rangle$~\cite{GP07b,GP07a} or the algebraic structure $\langle W,\mt,\jn,\circ,\id\rangle$ is a {\em distributive $\ell$-monoid}~\cite{GJ17,CGMS21}. Moreover, since time warps are join-preserving, there exist binary operations $\ld,\rd$ on $W$, called {\em residuals}, satisfying for all $f,g,h\in W$,
\[
f\le h\rd g \:\iff\: fg\le h \:\iff\: g\le f\ld h.
\]
That is, $\langle W,\mt,\jn,\circ,\ld,\rd,\id\rangle$ is a {\em residuated lattice}~\cite{BT03,MPT23}. From a programming perspective, residuals play a role similar to that of weakest preconditions in deductive verification. The time warp~$h \rd g$ may be viewed as the most general~(largest) time warp~$g'$ such that~$\Box_h A$ can be sent generically to~$\Box_g \Box_{g'} A$. In other words,~$\Box_{h \rd g} A$ is the most precise type~$B$ such that~$\Box_g B$ is a supertype of~$\Box_h A$. Similarly,~$f \ld h$ is the most general time warp~$f'$ such that~$\Box_h A$ can be sent generically to~$\Box_{f'} \Box_f A$. Such questions of genericity arise naturally when programming in a modular way~\cite{Guatto-2018}, justifying the consideration of residuals in gradings.

In order to benefit from the extra flexibility and descriptive power of adopting graded modalities in a programming language, the language implementation should be able to decide the order between gradings to distinguish between well-typed and ill-typed programs. That is, an algorithm is required that decides equations between terms interpreted over the algebraic structure of time warps. In a precursor to this paper~\cite{GGMS21}, decidability of the equational theory of the bounded residuated lattice $\langle W,\mt,\jn,\circ,\ld,\rd,\id,\bot,\top\rangle$, with least and greatest time warps $\bot$ and $\top$, respectively, was established. A key ingredient of the proof, however, involves translating terms into a rather awkward normal form which both complicates the decision procedure and obscures the algebraic meaning of the equations.

In this paper, following ideas of~\cite{Eklund2018,San20} on quantales of join-preserving maps on a complete lattice, we obtain an elegant normal form for time warp terms by introducing a constant for the  `predecessor' time warp $\dc$ that maps each $m\in\om{\setminus}\{0\}$ to $m-1$ and is constant on $\{0,\om\}$. More precisely, the {\em pointed} residuated lattice $\langle W,\mt,\jn,\circ,\ld,\rd,\id,\dc\rangle$ is term-equivalent to the involutive residuated lattice $\langle W,\mt ,\jn, \circ, \inv{},\id \rangle$, where $\inv{f}\coloneqq f\ld \dc$  is an \emph{involution} on the lattice $\langle L,\mt,\jn\rangle$. For convenience, we refer to both of these structures as the {\em time warp algebra} $\WarpA$. Because the monoid operation distributes over both joins and meets, and meets distribute over joins, every term is equivalent in $\WarpA$ to a meet of joins of {\em basic terms}, constructed using just the monoid operations and involution. It follows that the equational theory of $\WarpA$ is decidable if there exists an algorithm that decides $\WarpA\models \ut\le t_1\jn\cdots\jn t_n$ for arbitrary basic terms $t_1,\dots,t_n$, where $\ut$ is interpreted as the identity map. We provide such an algorithm by relating the existence of a counterexample to $\ut\le t_1\jn\cdots\jn t_n$ to the satisfiability of a corresponding first-order formula in $\som$, understood as an ordered structure with a decidable first-order theory, and describe an implementation, written in the OCaml programming language, that makes use of the Z3 theorem prover.


\subsection*{Overview of the paper}

In Section~\ref{sec:time-warp-algebra}, we introduce the time warp algebra $\WarpA$ as an example of an involutive residuated lattice consisting of  join-preserving maps on a complete chain (totally ordered set), and establish some of its elementary properties. In particular, we show that $\WarpA$ is simple, has no finite subalgebras, and generates a variety (equational class) that lacks the finite model property. We also describe a subalgebra $\m{R}$ of $\WarpA$ consisting of time warps that are `regular' in the sense that they are eventually either constant or linear, and provide an explicit description of the involution operation. 

In Section~\ref{sec:diagrams}, we describe the existence of a potential counterexample to $\ut\le t_1\jn\cdots\jn t_n$, where $t_1,\dots,t_n$ are basic terms, using the notion of a {\em diagram}, which provides a finite partial description of an evaluation of terms as time warps. As a consequence, an equation is satisfied by $\WarpA$ if, and only if, it is satisfied by $\m{R}$ (\Cref{thm:regular}). Finally, we use the resulting description to reduce the existence of a counterexample to the  satisfiability of a first-order formula in $\som$, understood as an ordered structure with a decidable first-order theory, thereby establishing the decidability of the equational theory of $\WarpA$ (\Cref{thm:decidable}). 

In Section~\ref{sec:implementation}, we describe an implementation of our decision procedure for the equational theory of $\WarpA$, written in the OCaml programming language~\cite{OCaml}, that makes use of the Z3 theorem prover~\cite{Z3}. Finally, in Section~\ref{sec:concluding}, we consider some potential avenues for further research; in particular, we explain how to adapt the decision procedure to deal with extra constants for first-order definable time warps such as $\bot$ and $\top$, and explore a relational approach to the study of the time warp algebra and related structures.


\section{The Time Warp Algebra}\label{sec:time-warp-algebra}

In this section, we introduce and establish some elementary properties of the time warp algebra $\WarpA$ in the general framework of involutive residuated lattices. In particular, we show that the variety generated by $\WarpA$ lacks the finite model property and that  `regular' time warps---those that are eventually either constant or linear---form a subalgebra of $\WarpA$. We also provide an explicit description of the involution operation on time warps.

A  {\em pointed residuated lattice} (also known as an {\em FL-algebra}) is an algebraic structure $\m{L}=\langle L,\mt,\jn,\cdot,\ld,\rd,\ut,\zr\rangle$ of signature $\langle 2,2,2,2,2,0,0\rangle$ such that $\langle L,\pd,\ut \rangle$ is a monoid, $\langle L,\mt,\jn\rangle$ is a lattice with an order defined by $a \le b \defiff a \mt b = a$, and, for all $a,b,c \in L$,
\[
a \le c \rd b \:\iff\: a b \le c \:\iff\: b \le a \ld c,
\]
where the residual operators are given explicitly for $a,b\in L$ by
\[
a\ld b=\bigvee\{c\in L\mid ac\le b\}
 \quad \text{and} \quad
b\rd a=\bigvee\{c\in L\mid ca\le b\}.
\]
Moreover, it follows from the existence of residual operations $\ld$ and $\rd$ that the operation $\pd$ preserves existing joins in both coordinates; in particular, for all $a,b_1,b_2,c\in L$,
\[
a(b_1\jn b_2)c=ab_1c\jn ab_2c.
\]
The structure $\m{L}$ is said to be {\em distributive} if its lattice reduct $\langle L,\mt,\jn\rangle$ is distributive and {\em fully distributive} if it is distributive and, for all $a,b_1,b_2,c\in L$,
\[
a(b_1\mt b_2)c=ab_1c\mt ab_2c.
\]
The element $\zr\in L$ is called \emph{cyclic} if $\zr\rd a = a \ld\zr$, for all $a \in L$, and \emph{dualizing} if, for all $a \in L$,
\[
\zr\rd (a \ld\zr) = a = (\zr\rd a) \ld \zr.
\]
If $\zr$ is both cyclic and dualizing, then the map $'\colon L \to L;\:x \mapsto x\ld\zr$ is an \emph{involution} on the lattice $\langle L, \mt, \jn \rangle$, satisfying for all $a, b\in L$, 
\[
a \leq b \iff b' \leq a',\quad a'' = a,\quad (a \mt b)' = a' \jn b', \quad \text{and} \quad (a \jn b)' = a' \mt b'.
\]
The class of pointed residuated lattices such that $\zr$ is cyclic and dualizing forms a variety (equational class) that is term-equivalent to the variety of  \emph{involutive residuated lattices}: algebraic structures $\langle L,\mt,\jn,\pd,\mathop{'},\ut \rangle$ of signature $\langle 2, 2, 2, 1, 0 \rangle$  such that  $\langle L, \pd,\ut \rangle$ is a monoid, $\langle L, \mt, \jn \rangle$ is a lattice, $'$ is an involution on $\langle L,\mt ,\jn \rangle$, and for all $a,b,c \in L$,
\[
b \leq (c'a)'  \:\iff\: ab \leq c\:\iff\: a \leq (b c')'.
\]
The term-equivalence is implemented by defining $x' \coloneqq x\ld\zr$ and, conversely, $x\ld y \coloneqq (y'x)'$, $y\rd x \coloneqq (xy')'$, and $\zr\coloneqq\ut'$ (see~\cite{MPT23} for further details). For any element $a$ of an involutive residuated lattice, we also define inductively $a^0:=\ut$ and $a^{k+1}:=a^k\pd a$ ($k\in\mathbb{N}$).

Now let $\m{Tm}$ denote the {\em term algebra} of the language of involutive residuated lattices defined over a countably infinite set of variables ${\rm Var}$, and call a term $t\in{\rm Tm}$ {\em basic} if it is constructed using the operation symbols $\pd$, $\mathop{'}$, and $\ut$. As usual, an {\em equation} is an ordered pair of terms $s,t\in{\rm Tm}$, denoted by $s\eq t$, and $s\leq t$ abbreviates $s\mt t\eq s$. An involutive residuated lattice $\m{L}$ {\em satisfies} an equation $s\eq t$, denoted by $\m{L}\models s\eq t$, if $\sem{s}_h=\sem{t}_h$ for every homomorphism $h\colon\m{Tm}\to\m{L}$, where $\sem{u}_h:=h(u)$  for  $u\in{\rm Tm}$. 

In a fully distributive involutive residuated lattice, every term is equivalent to a meet of joins of basic terms. More precisely:

\begin{lem}\label[lem]{lem:normal}
There exists an algorithm that produces for any term $t\in{\rm Tm}$ positive integers $m,n_1,\dots,n_m$ and basic terms $t_{i,j}$ for each $i\in\{1,\dots,m\}$ and $j\in\{1,\dots,n_i\}$ such that for any fully distributive involutive residuated lattice $\m{L}$,
\[
\m{L}\models t \eq \bigmt_{i=1}^m \bigjn_{j=1}^{n_i} t_{i,j}.
\]
\end{lem}
\begin{proof}
The desired basic terms are obtained by iteratively distributing joins over meets and the monoid multiplication over both meets and joins, and pushing the involution inwards using the De Morgan laws.
\end{proof}

Next, let $\m{C}$ be any complete chain with a least element $0$ and a greatest element $\infty$. The set $\res(\m{C})$ of maps on $\m{C}$ that preserve arbitrary joins forms a fully distributive pointed residuated lattice $\Res(\m{C})=\langle \res(\m{C}),\mt,\jn,\circ,\ld,\rd,\id,\dc \rangle$, where $\mt,\jn$ are defined pointwise, $\circ$ is functional composition, $\id$ is the identity map, and $\dc$ is the join-preserving map
\[
\dc\colon C\to C;\quad x\mapsto\bigjn \{y \in C \mid y < x\}.
\]
The least and greatest elements of $\Res(\m{C})$ are the maps $\bot,\top$ satisfying, respectively, $\bot(x) = 0$, for all $x\in C$, and $\top(0) = 0$ and $\top(x) = \infty$ for all $x\in C{\setminus}\{0\}$. Moreover, it follows from~\cite[Proposition~2.6.18]{Eklund2018} and~\cite[Section~4]{San20}) that $\dc$ is the unique cyclic and dualizing element of $\Res(\m{C})$ and hence that $\Res(\m{C})$ is term-equivalent to the involutive residuated lattice $\langle \res(\m{C}),\mt,\jn, \circ,\inv{},\id\rangle$, where $\inv{f}\coloneqq f\ld \dc = \dc\rd f$ for each $f\in\res(\m{C})$.

Recall that an algebraic structure $\m{A}$ is \emph{simple} if $\mathrm{Con}(\m{A}) = \{\De_A,\nabla_A\}$, where $\mathrm{Con}(\m{A})$ denotes the set of congruences of $\m{A}$,  $\De_A=\{\langle a,a\rangle \mid a\in A\}$, and $\nabla_A = A\times A$. 

\begin{prop}\label[prop]{prop:simple}
For any complete chain $\m{C}$, the pointed residuated lattice $\Res(\m{C})$ is simple.
\end{prop}
\begin{proof}
Consider any $\Theta\in\mathrm{Con}(\Res(\m{C})){\setminus}\{\Delta_{\res(\m{C})}\}$. Since $\Theta\neq\Delta_{\res(\m{C})}$ and a congruence of a pointed residuated lattice is fully determined by the congruence class of its multiplicative unit (see~\cite{BT03}), there exists an $f\in\res(\m{C}){\setminus}\{\id\}$ such that $f\mathrel{\Theta}\id$. Choose any $c \in C$ such that $f(c) \neq c$. We consider the cases $f(c) < c$ and $c < f(c)$. Suppose first that $f(c) < c$ and define the map $g\in\res(\m{C})$ such that for each $d\in C$,
\[
g(d) := \begin{cases}
0 &\text{if } d\leq f(c); \\
\infty &\text{if } d>f(c).
\end{cases}
\]
Since $f\mathrel{\Theta}\id$ and $\Theta$ is a congruence, also $((g\circ f) \rd g) \mt \id\mathrel{\Theta}((g\circ \id) \rd g)\mt\id$. Observe that
\[
((g \circ f) \rd g)(\infty) = ((g\circ f) \rd g)(g(c)) =(((g\circ f) \rd g)\circ g)(c) \leq (g \circ f)(c) = g(f(c)) = 0,
\]
so $\bot = (g\circ f) \rd g =  ((g\circ f) \rd g) \mt \id$. Also, $((g\circ \id) \rd g) \mt \id = (g \rd g) \mt \id  = \id$, so $\bot\mathrel{\Theta}\id$ and  $\top=\bot\ld\id\mathrel{\Theta} \id \ld \id =  \id$. Hence $\top\mathrel{\Theta}\bot$. But $\Theta$ is a lattice congruence, so its congruence classes are convex and $\Theta=\nabla_{\res(\m{C})}$. Suppose next that $c<f(c)$. Then $f((f\ld \id)(c)) = (f\circ(f\ld\id))(c)\leq c < f(c)$, and, since $f$ is order-preserving and $\m{C}$ is a chain, $(f\ld id)(c) < c$. But then, since  $f\ld \id \mathrel{\Theta}\id \ld \id = \id$,  as in the previous case, $\Theta=\nabla_{\res(\m{C})}$.
\end{proof}

Let us focus our attention now on the special case of the successor ordinal $\som=\om\cup\{\om\}$. For convenience, we call both the pointed residuated lattice $\Res(\som)$ and the corresponding term-equivalent involutive residuated lattice, the {\em time warp algebra} $\m{W}$, referring to members of $W$, i.e., join-preserving maps on $\som$, as \emph{time warps}. Clearly, a map $f\colon\som\to\som$ is a time warp if, and only if, it is order-preserving and satisfies $f(0) = 0$ and $f(\om) = \bigjn\{f(n) \mid n\in \om\}$. As motivation for subsequent sections, we will show below that equational reasoning in $\m{W}$ cannot be checked by considering finite members of the variety generated by this algebra.

Observe first that $\dc:=\inv{\id}\in W$ is the `predecessor' time warp satisfying for  $m\in\som$,
\[
\dc(m) =  \bigjn \{n \in\om \mid n < m\} =  \begin{cases}
m & \text{if } m \in \{0,\om \};\\
m-1 & \text{otherwise.}
\end{cases}
\]
Clearly, $\dc<\id$, so $\WarpA$ satisfies the equation $\ut' \leq\ut$.

The involution operation can be described using $\dc$ as follows.

\begin{lem}\label[lem]{lem:invf}
For any time warp $f$ and $m\in \som$,
\[
\inv{f}(m) = \bigmt \{ \dc(n) \mid n \in \som\text{ and }\, m \leq f(n)\}.
\]
\end{lem}
\begin{proof}
Let $h$ be the function defined by $h(m) \defeq\bigmt \{ \dc(n) \mid n \in \som \text{ and }\,m \leq f(n)\}$ for each $m\in\som$. Clearly, $h$ is order-preserving and satisfies $h(0)=0$. Hence, to show that $h$ is a time warp, it remains to check that $h(\om)=\bigvee\{h(m)\mid m\in\om\}$. To this end, suppose first  that $h(\om) = \om$. Then $f(n) < \om$, for all $n \in \om$. If also $f(\om) < \om$, then $h(m) = \om$,  for each $m > f(\om)$, yielding $h(\om)= \om = \bigvee\{h(m)\mid m\in\om\}$. If $f(\om) = \om$, then, since $\om = f(\om) = \bigjn \{ f(n) \mid n \in \om\}$, there exists, for each $m \in \om$, an $n\in \om$ with $m\leq f(n) < \om$, yielding again $h(\om)= \om = \bigvee\{h(m)\mid m\in\om\}$. Suppose now that $h(\om) < \om$. Then there exists a minimal $n\in \om{\setminus} \{0\}$ such that $f(n) = \om$, and hence $h(m) = p(n)$, for each $m > f(n-1)$, yielding $h(\om)= p(n) = \bigvee\{h(n)\mid n\in\om\}$.
 
 Moreover, $hf\le \dc$, so $h\le \dc\rd f =  \inv{f}$. Now consider any $m\in \om{\setminus}\{0\}$. If $n\in\som$ satisfies $m\le f(n)$, then $\inv{f}(m)\leq \inv{f}(f(n)) \le \dc(n)$. So $\inv{f}(m)\le  h(m)$. Hence $h = \inv{f}$, since time warps are determined by their values on $\om{\setminus}\{0\}$.
\end{proof}

Observe next that $p^{k+1} < p^k$ for each $k\in\mathbb{N}$ and hence that the $\emptyset$-generated subalgebra of $\WarpA$ is infinite. So $\WarpA$ has no finite subalgebras. Moreover, $\WarpA$ is simple, by \Cref{prop:simple}, so any quotient of $\WarpA$ is either trivial or isomorphic to $\WarpA$. Indeed, we do not know if the variety generated by $\WarpA$ contains {\em any} non-trivial finite algebra, but can show at least that it cannot contain $\Res(\m{C})$ for any finite chain $\m{C}$ and does not have the finite model property.\footnote{Note that the involution-free reduct of $\WarpA$ generates the variety of {\em distributive $\ell$-monoids}, which does have the finite model property; indeed, an equation is satisfied by this reduct of $\WarpA$ if, and only if, it is satisfied by the distributive $\ell$-monoid of join-preserving maps on any finite chain~\cite{CGMS21}.}

To this end, observe first that the element $\dc$ has a right inverse; that is, $\dc \circ s = \id$, where $s$ is the `successor' time warp satisfying for each $m\in\som$,
\[
s \colon \som \to \som;\quad
m\mapsto \begin{cases}
m & \text{if } m \in \{ 0, \om\};\\
m+1 & \text{otherwise.}
\end{cases}
\]
Note that $s = \dc \ld \id$, since $\dc \ld \id = \inv{( \inv{\id} \circ \dc)} = \inv{(\dc\circ \dc)}$ and, by \Cref{lem:invf}, for each $m\in \som$,
\begin{align*}
\inv{(\dc\circ \dc)}(m) 
& = \bigmt \{ \dc(n) \mid n \in \som\text{ and }\, m \leq (\dc\circ \dc)(n)\}\\
& = \bigmt \{ n \mid n \in \som\text{ and }\, m \leq \dc(n)\}\\
& = s(m).
\end{align*}
Hence $\dc\circ (\dc \ld \id) = \id$, and $\WarpA$ satisfies the equation $\ut\eq\ut' \pd (\ut'\ld\ut)$.

\begin{prop}
The variety generated by $\WarpA$ does not have the finite model property and does not contain $\Res(\m{C})$ for any finite chain $\m{C}$.
\end{prop}
\begin{proof}
Let $\langle L,\mt,\jn,\pd,\mathop{'},\ut \rangle$ be any finite member of the variety generated by $\WarpA$. Then, since $\ut' \leq \ut$ and $\bf L$ is finite, there exists a $k \in \mathbb{N}$ such that $(\ut')^k = (\ut')^{k+1}$. Moreover, $\ut'=\ut'\ut=(\ut')^2\pd(\ut'\ld\ut)$, since $\ut=\ut'\pd(\ut'\ld\ut)$ and $\ut$ is the multiplicative unit of $\bf L$, and hence, iterating this step, 
\[
\ut' = (\ut')^{k+1}(\ut' \ld \ut)^k = (\ut')^k(\ut' \ld \ut)^k =  \ut.
\]
The equation $\ut'\eq\ut$ is therefore satisfied by all the finite members of the variety generated by $\WarpA$, but not by $\WarpA$ itself, since $\dc < \id$. Moreover, since $\Res(\m{C})$ is finite for any finite chain $\m{C}$ and does not satisfy $\ut'\eq\ut$, it cannot belong to  the variety generated by $\WarpA$.
\end{proof}

Although the variety generated by $\WarpA$ does not have the finite model property, we will show in this paper that it is generated by a subalgebra of $\WarpA$ consisting of time warps that have a simple finite description. Let us call a time warp $f\in W$ {\em eventually constant}  if there exists an $m\in\om$ such that $f(n)=f(m)$ for all $n\in\om$ with $n\ge m$,  {\em eventually linear} if  there exist an $m\in\om$ and a $k\in\mathbb{Z}$ such that $f(n)=n+k$ for all $n\in\om$ with $n\ge m$, and {\em regular} if it is eventually constant or eventually linear. Equivalently, a time warp $f\in W$ is regular if, and only if, there exist $m\in\om$, $l\in\{0,1\}$, and $k\in\mathbb{Z} \cup \{ \om \}$ such that $f(n)=ln+k$ for all $n\in\om$ with $n\ge m$.

\begin{prop}\label[prop]{prop:regular-subalgebra}
The set of regular time warps forms a subalgebra $\m{R}$ of $\WarpA$.
\end{prop}
\begin{proof}
Clearly, $\id$ is eventually linear, and hence regular. Suppose that $f,g\in W$ are regular. It is easy to see that $f\mt g$ and $f\jn g$ are then also regular. If $g$ is eventually constant, then there exists an $m\in\om$ such that $g(n)=g(m)$ for all $n\in\om$ with $n\ge m$, and hence also $f(g(n))=f(g(m))$ for all $n\in\om$ with $n\ge m$, so $f\circ g$ is eventually constant. Suppose then that $g$ is eventually linear, that is,  there exist an $m\in\om$ and a $k\in\mathbb{Z}$ such that $g(n)=n+k$ for all  $n\in\om$ with $n\ge m$. If $f$ is eventually constant, then there exists an $l\in\om$ such that $f(n)=f(l)$ for all $n\in\om$ with $n\ge l$, and hence $f(g(n))=f(l)$ for all $n\in\om$ with $n\ge \max(m,l-k)$, that is, $f\circ g$ is eventually constant. If $f$ is eventually linear, then  there exist an $l\in\om$ and a $r\in\mathbb{Z}$ such that $f(n)=n+r$ for all  $n\in\om$ with $n\ge l$, and hence $f(g(n))=n+k+r$ for all $n\in\om$ with $n\ge \max(m,l+k)$. 

It remains to show that $\inv{f}$ is regular for $f$ regular. Suppose first that $f$ is eventually constant, that is, there exists an $m\in\om$ such that $f(n)=f(m)$ for all $n\in\om$ with $n\ge m$.  If $f(m) < \om$, then for each $k\in \om$ with $k\ge f(m)+1$, the set $\{ \dc(n) \mid n \in \som \text{ and } k \leq f(n) \}$ is empty and $\inv{f}(k) = \om$, by \Cref{lem:invf}, i.e., $\inv{f}$ is eventually constant. Otherwise, $f(m) = \om$ and  for all $l\geq 1 + \max \{ f(n) \mid n\in \om,  f(n) < \om \}$,  \Cref{lem:invf} yields
\[
\inv{f}(l) =\bigmt\{ \dc(n) \mid n \in \som \text{ and }\, l \leq f(n) \} =\bigmt\{ \dc(n) \mid n \in \som \text{ and }\, f(n) = \om \},
\]
i.e., $\inv{f}$ is eventually constant. Finally, suppose that $f$ is eventually linear, that is, there exists an $m \in \om$ and a $k\in \mathbb{Z}$ such that $f(n) = n+k$ for all $n\in \om$ with $n\geq m$. Then every $l \in \om$ with $l\geq m+k +1$ lies in the image of $f$ and  $n = l - k$ is the unique solution for the equation $l = f(n)$, so $\inv{f}(l) = l - (k+1)$, by \Cref{lem:invf}, i.e., $\inv{f}$ is eventually linear.
\end{proof}

We conclude this section by providing an explicit description of the involution operation on time warps that will play a crucial role in the decidability proof in the next section. For any time warp $f\in W$, let $\lastw(f)$ denote the smallest $m\in\som$ such that $f$ takes the same value on $m$ (and hence all elements greater than $m$) as $\om$. More formally:
\[
\lastw(f) \coloneqq \min \{m \in \som \mid f(m) = f(\om)\}.
\]
Observe that $\lastw(f)<\om$ if, and only if, $f$ is eventually constant. Moreover, we have $\lastw(f)=(id\rd f)f(\om)$, since $\id \rd f = \inv{(f \circ \inv{\id})} = \inv{(f\circ \dc)}$ and, by \Cref{lem:invf}, 
\[
\inv{(f\circ \dc)}(f(\om)) = \bigmt\{ \dc(n) \mid n \in \som \text{ and }\, f(\om) \leq f(\dc(n)) \}  = \min \{m \in \som \mid f(m) = f(\om)\}.
\]
The behaviour of this operation with respect to compositions and involutions of time warps is easily described as follows.

\begin{lem}\label[lem]{lemma:last-properties}
For any time warps $f$ and $g$,
\begin{align*}
\lastw(fg) = \om & \:\iff\: \lastw(g) = \lastw(f) = \om\\
\lastw(\inv{f}) = \om & \:\iff\: \lastw(f) = \om.
\end{align*}
\end{lem}
\begin{proof}
For the first equivalence, observe that $\lastw(fg)=\om$ if, and only if, $fg(m)<fg(\om)$ for all $m\in\om$. However, $fg(m)=fg(\om)$ for some $m\in\om$ if, and only if, either $g(m)=g(\om)$ for some $m\in\om$, or $g(m)<g(\om)$ for all $m\in\om$ and $f(k)=f(\om)$ for some $k\in\om$, which is equivalent to $\lastw(g)<\om$ or $\lastw(f)<\om$. For the second equivalence, observe that $\lastw(\inv{f}) = \om$ if, and only if, $\inv{f}(m)<\inv{f}(\om)$ for all $m\in\om$. But $\inv{f}(m)=\inv{f}(\om)$ for some $m\in\om$ if, and only if, $f(m)=f(\om)$ for  some $m\in\om$, by Lemma~\ref{lem:invf}, which is equivalent to $\lastw(f)<\om$.
\end{proof}

\noindent
Finally, we are able to provide the promised explicit description of the involution operation on time warps.

\begin{prop}\label[prop]{prop:inv}
For any time warp $f$, $n\in \om{\setminus}\{0\}$, and $m\in \om$,
\begin{align*}
\inv{f}(n) = m &\:\iff\:  f(m) < n \leq f(m+1) \\
\inv{f}(\om) = m &\:\iff\:  f(m) < \om = f(m+1) \\
\inv{f}(n) = \om &\:\iff\: f(\om) < n \\
\inv{f}(\om) = \om &\:\iff\: f(\om) < \om\,\text{ or }\,\lastw(f) = \om.
\end{align*}
\end{prop}
\begin{proof}
Consider any $n\in\som{\setminus}\{0\}$ and $m\in\om$. For the first two equivalences, observe that, using \Cref{lem:invf} and the assumption that  $n\neq 0$,
\[
\inv{f}(n) = m \:\iff\: m = \bigmt\{\dc(k) \mid k \in \som \text{ and }\,n\leq f(k)\} \:\iff\: f(m) <n \leq f(m+1).
\]
For the second two equivalences, observe first that, using  \Cref{lem:invf},
\[
\inv{f}(n) = \om \:\iff\: \om = \bigmt \{ \dc(k) \mid k \in \som \text{ and }\,n \leq f(k) \} \:\iff\: f(k) < n\:\text{ for all $k\in \om$}.
\]
It follows that $\inv{f}(n) = \om \iff f(k) < n\,\text{ for all $k\in \om$}\iff f(\om) < n$, for any $n\in\om{\setminus}\{0\}$, and  $\inv{f}(\om) = \om \iff  f(k) < \om\:\text{ for all $k\in \om$}\iff f(\om) < \om\,\text{ or }\lastw(f) = \om$. \qedhere
\end{proof}


\section{Decidability via Diagrams}\label{sec:diagrams}

In this section, we turn our attention to the problem of deciding equations in the time warp algebra $\WarpA$. Observe first that for any equation $s\eq t$ with $s,t\in{\rm Tm}$, 
\[
\WarpA\models s\eq t
\:\iff\:
(\WarpA\models s\le t\,\text{ and }\,\WarpA\models t\le s)
\:\iff\:
\WarpA\models\ut\le(s\ld t)\mt(t\ld s).
\]
Moreover, by \Cref{lem:normal}, each $u\in{\rm Tm}$ is equivalent in $\WarpA$ to a meet of joins of basic terms. Hence, since $\WarpA\models\ut\le u_1\mt\cdots\mt u_m$ if, and only if, $\WarpA\models\ut\le u_i$, for each $i\in\{1,\dots,m\}$, we may restrict our attention to deciding equations of the form 
$\ut\le t_1\jn\cdots\jn t_n$, where  $\{t_1,\dots,t_n\}$ is any finite non-empty set of basic terms. More precisely:

\begin{prop}\label{prop:simpleform}
The equational theory of $\WarpA$ is decidable if, and only if, there exists an algorithm that decides $\WarpA\models \ut\le t_1\jn\cdots\jn t_n$ for any basic terms $t_1,\dots,t_n$.
\end{prop}

\noindent
To address this problem, we relate the existence of a counterexample to $\WarpA\models\ut\le t_1\jn\cdots\jn t_n$, where $t_1,\dots,t_n$ are basic terms, to the satisfiability of a corresponding first-order formula in $\som$, understood as an ordered structure with a decidable first-order theory. Such a counterexample is given by assigning time warps to the variables in $t_1,\dots,t_n$ to obtain time warps $\hat{t}_1,\dots,\hat{t}_n$ satisfying $\id\not\le\hat{t}_1\jn\cdots\jn \hat{t}_n$, that is, by a homomorphism $h\colon\m{Tm}\to\WarpA$ and element $k\in\som$ satisfying $k>\sem{t_i}_h(k)$, for each $i\in\{1,\dots,n\}$. The translation into a first-order formula is obtained in two steps. First, a `time variable' $\ka$ is introduced to stand for the unknown $k\in\som$ and finitely many `samples' are generated that correspond to other members of $\som$ used in the computation of $\sem{t_1}_h(k),\dots,\sem{t_n}_h(k)$. Second, finitely many quantifier-free formulas are defined that describe the relationships between samples according to the semantics of $\WarpA$. The required formula is then the existential closure of the conjunction of these quantifier-free formulas and a formula that expresses $\ka>\sem{t_i}_h(\ka)$ for each $i\in\{1,\dots,n\}$.

Let us fix a countably infinite set $\TermV$ of \emph{time variables}, denoted by the (possibly indexed) symbol $\ka$. We call any member of the language generated by the following grammar a \emph{sample}, where $t\in{\rm Tm}$ is any basic term and $\ka\in\TermV$ is any time variable:
  \[
    \Time \ni
    \al
    \Coloneqq
    \ka \mid \s{t}{\al} \mid \suc(\al) 
    \mid \last(t).
  \] 
Now let $\leadsto$ be the binary relation defined on the set of all samples by
  \begin{align*}
    \s{t}{\al}
    & \leadsto
    \al
    &
    \s{tu}{\al}
    & \leadsto 
    \s{t}{\s{u}{\al}}
    \\
    \suc(\al)
    & \leadsto
    \al
    &
      \s{\invt{t}}{\al}
    & \leadsto
    \s{t}{\s{\invt{t}}{\al}}
    \\
      \s{t}{\al}
    & \leadsto
    \s{t}{\last(t)}
   &    \s{\invt{t}}{\al}
    & \leadsto
    \s{t}{\suc(\s{\invt{t}}{\al})},
  \end{align*}
and let  $\leadsto^*$ denote the reflexive transitive closure of this relation. We call a sample set (i.e., set of samples) $\De$ \emph{saturated} if whenever~$\al \in \De$ and~$\al \leadsto \be$, also~$\be \in \De$, and define the~\emph{saturation} of a sample set~$\De$ (analogously to the Fischer-Ladner closure of formulas in Propositional Dynamic Logic \cite{Fischer1979}) as
\[
\sat{\De}\defeq\{\be\in\Time\mid \al \leadsto^* \be\text{ for some }\al \in \De\}.
 \]
 Crucially, this definition yields the following property:

\begin{lem}\label[lem]{l:fin-sat}
The saturation of a finite sample set is finite.
\end{lem}
\begin{proof}
It suffices to consider a sample set with one element, so suppose that $\De = \{\al_0\}$. Let $K = \{ \ka_\al \mid \al \in \Time\}\subseteq\TermV$ be the set of time variables not occurring in $\al_0$ indexed by the set  of all samples $\Time$, and consider a further binary relation $\altsat$ defined on $\Time$ by
 \begin{align*}
    \s{t}{\al}
    & \altsat
    \al
    &
    \s{tu}{\al}
    & \altsat 
    \s{t}{\s{u}{\al}}
    \\
    \suc(\al)
    & \altsat
    \al
    &
      \s{\invt{t}}{\al}
    & \altsat
    \s{t}{\ka_{\s{\invt{t}}{\al}}}
    \\
      \s{t}{\al}
    & \altsat
    \s{t}{\last(t)}
    &
      \s{\invt{t}}{\al}
    & \altsat
  \s{t}{\ka_{\suc(\s{\invt{t}}{\al})}}.
  \end{align*}
Now let $\De^\altsat \defeq \{\be\in\Time \mid \al \altsat^* \be \text{ for some }\al \in \De\}$, where $\altsat^*$ is the reflexive transitive closure of $\altsat$. For a sample $\al$, denote by $v(\al)$ the set of variables from ${\rm Var}$ that occur in $\al$, and by  $\ell(\al)$ the number of occurrences in $\al$ of variables from ${\rm Var}$ and symbols in $\{\cdot, \invt{ \ },\ut, \suc \}$ that do not occur in the scope of an occurrence of $\last$. Note that, by definition, $\altsat$ does not increase $v$ and $\ell$, i.e., if $\al \altsat \be$, then  $v(\be) \subseteq v(\al)$, $ \ell(\be)\leq \ell(\al)$, and if $\al \altsat \be$ is not of the  form  $\s{t}{\ga} \altsat \s{t}{\last(t)}$, then $\ell(\be) <\ell(\al)$. In particular,  $v(\be) \subseteq v(\al_0)$ and $\ell(\be) \leq \ell(\al_0)$, for each $\be \in \De^{\altsat}$. But clearly there are only finitely many samples of the form $\s{t}{\last(t)}$ such that $v(\s{t}{\last(t)}) \subseteq v(\al_0)$ and $\ell(\s{t}{\last(t)}) \leq \ell(\al_0)$.  Hence $\De^{\altsat}$ is finite.
  
Now define the map $\si\colon \De^{\altsat} \to \sat{\De}$ recursively by $\si(\ka) = \ka$ for $\ka \in \TermV\setminus K$, $\si(\last(t)) = \last(t)$,  $\si(\ka_\al) = \si(\al)$ for $\ka_\al \in K$,  $\si(\suc(\al)) = \suc(\si(\al))$, $\si(\s{t}{\al}) = \s{t}{\si(\al)}$. It is straightforward to check that $\si$ is well-defined and surjective. So $\sat{\De}$ is finite.
\end{proof} 

\begin{runexa}
Let us consider the equation $\ut \leq x\invt{x}$ as a running example throughout this section. In this case, the sample set of interest is $\{ \s{x \invt{x}}{\ka} \}$. In \Cref{fig:1}, its saturation  $\sat{\{ \s{x \invt{x}}{\ka} \}}$ is visualized as a tree, where each parent node is related to its successors by $\leadsto$ and redundant samples are omitted. 
\begin{figure}
\centering
\begin{forest} for tree={
    edge path={\noexpand\path[\forestoption{edge}] (\forestOve{\forestove{@parent}}{name}.parent anchor) -- +(0,-12pt)-| (\forestove{name}.child anchor)\forestoption{edge label};}
}
[$\s{x\invt{x}}{\ka}$
 [$\s{x}{\s{\invt{x}}{\ka}}$
  [$\s{\invt{x}}{\ka}$ 
   [ $\ka$ ] 
   [$\s{x}{\suc(\s{\invt{x}}{\ka})}$ [ $\suc(\s{\invt{x}}{\ka})$ ] ]
   [$\s{\invt{x}}{\last(\invt{x})}$ 
    [ $\last(\invt{x})$ ] 
    [$\s{x}{\s{\invt{x}}{\last(\invt{x})}}$]
    [$\s{x}{\suc(\s{\invt{x}}{\last(\invt{x})})}$ [ $\suc(\s{\invt{x}}{\last(\invt{x})})$ ] ]
   ]
  ]
  [$\s{x}{\last(x)}$ [$\last(x)$] ]
 ]
 [$\s{x\invt{x}}{\last(x\invt{x})}$
  [$\s{x}{\s{\invt{x}}{\last(x\invt{x})}}$
   [$\s{\invt{x}}{\last(x\invt{x})}$
   [$\last(x\invt{x})$]
    [$\s{x}{\suc(\s{\invt{x}}{\last(x\invt{x})})}$ [ $\suc(\s{\invt{x}}{\last(x\invt{x})})$ ] ]
   ]
  ]
 ]
]
\end{forest}
\caption{Visualization of the saturation of  $ \{\s{x\invt{x}}{\ka} \}$}
\label{fig:1}
\end{figure}
\end{runexa}

Consider next a first-order signature $\tau = \{\aleq, \aS, 0, \om \}$  with a binary relation symbol $\aleq$, a unary function symbol $\aS$, and two constants $0$ and $\om$. We denote by $\fosom$ the $\tau$-structure with universe $\som$ and natural order $\preceq^\fosom$, defining $\aS^\fosom(n) := n+1$, for each $n\in \om$, $\aS^\fosom(\om) := \om$, $\om^\fosom := \om$, and $0^\fosom := 0$. Since no other $\tau$-structure will be considered in this section, we will omit the superscripts from now on. We use the symbols $\neg$, $\andd$, $\orr$, $\To$, and $\Equiv$ to denote the logical connectives `not', `and', `or', `implies', and `if, and only if', respectively, and let $a \ale b$ stand for $a \aleq b \andd \neg(b \aleq a)$.

Let us fix now a saturated sample set $\De$ and consider its members as first-order variables. We define the following sets of first-order quantifier-free $\tau$-formulas over $\De$:
\begin{align}
\mathsf{struct}_\De \coloneqq \; & \{\al \aleq \be \To \s{t}{\al} \aleq \s{t}{\be} \mid \s{t}{\al}, \s{t}{\be} \in \De \} \; \cup 
\label{c:mon}
\\
&\{\al \eq 0 \To \s{t}{\al} \eq 0 \mid \s{t}{\al}  \in \De \}\;  \cup 
\label{c:zero}
\\
&\{\suc(\al) \eq \aS(\al) \mid \suc(\al) \in \De \} \; \cup 
\label{c:suc}
\\
&\{ \last(t) \aleq \al \Equiv \s{t}{\last(t)} \eq \s{t}{\al} \mid \s{t}{\al} \in \De \} \; \cup
\label{c:last}
\\
&\{\last(t) \eq \om \To \s{t}{\last(t)} \eq \om \mid \s{t}{\last(t)} \in \De \}
\label{c:last2}
\\[5pt]
\mathsf{log}_\De \coloneqq \; & \{ \s{\ut}{\al} \approx \al \mid \s{\ut}{\al} \in \De \} \;  \cup 
\label{c:id}
\\
&\{\last(\ut) \eq \om \mid \last(\ut) \in \De \} \;   \cup 
\label{c:id2}
\\
&\{ \s{tu}{\al} \eq \s{t}{\s{u}{\al}} \mid \s{tu}{\al} \in \De\} \;  \cup 
\label{c:prod}
\\
&\{\last(tu) \eq \om \To (\last(t) \eq \om \andd  \last(u) \eq \om) \mid \last(tu),\last(t),\last(u) \in \De  \}
\label{c:last-prod1}
\\[5pt]
\mathsf{inv}_\De \coloneqq \; & \{(0 \ale \al \andd \al \ale \om) \To \s{t}{\s{\invt{t}}{\al}} \ale \al \mid \s{\invt{t}}{\al} \in \De  \} \; \cup 
\label{c:inv-lower}
\\
&\{\s{\invt{t}}{\al} \ale \om \To \al \aleq  \s{t}{\suc(\s{\invt{t}}{\al})} \mid \s{\invt{t}}{\al} \in \De   \} \; \cup 
\label{c:inv-upper}
\\
&\{\last(\invt{t}) \eq \om \To \last(t) \eq \om \mid \last(\invt{t}), \last(t) \in \De \}
\label{c:inv-last-inf}
\\
\Si_\De \coloneqq \; &  \mathsf{struct}_\De \cup \mathsf{log}_\De \cup \mathsf{inv}_\De.
\end{align}
We call a $\fosom$-valuation $\de \colon \De \to \som$ a \emph{$\De$-diagram} if $\fosom, \de\models \Si_\De$.\footnote{The term `diagram' recalls a similar concept used to prove the decidability of the equational theory of lattice-ordered groups in~\cite{HM79}.} 

\begin{runexa}\label{diagram}
For $\De =  \sat{\{ \s{x \invt{x}}{\ka} \}}$, consider the map $\de \colon \De \to \som$ defined by
\begin{align*}
0 &= \de(\s{\invt{x}}{\ka}) = \de(\s{x}{\s{\invt{x}}{\ka}}) = \de(\s{x\invt{x}}{\ka}) , \\
1 &= \de(\s{x}{\suc(\s{\invt{x}}{\ka})}) = \de(\suc(\s{\invt{x}}{\ka})) = \de(\ka), \\
\om &= \de(\last(x)) = \de(\last(\invt{x})) = \de(\last(x\invt{x})) = \de(\s{x}{\last(x)}) =  \de(\s{\invt{x}}{\last(\invt{x})}) \\
&= \de(\s{x}{\s{\invt{x}}{\last(\invt{x})}}) = \de(\s{x}{\suc(\s{\invt{x}}{\last(\invt{x})})}) = \de(\suc(\s{\invt{x}}{\last(\invt{x})})) = \de(\s{x\invt{x}}{\last(x\invt{x})})  \\
&= \de(\s{x}{\s{\invt{x}}{\last(x\invt{x})}}) = \de(\s{\invt{x}}{\last(x\invt{x})}) = \de(\s{x}{\suc(\s{\invt{x}}{\last(x\invt{x})})}) = \de(\suc(\s{\invt{x}}{\last(x\invt{x})})).
\end{align*}
It is readily checked that $\de$ is a $\De$-diagram; e.g., $\de$ satisfies \Cref{c:zero}, since $\de(\s{\invt{x}}{\ka}) = 0$  and $\de(\s{x}{\s{\invt{x}}{\ka}}) = 0$.
\end{runexa}

Fixing a set of basic terms $T$ and time variable $\ka$, the following proposition relates any homomorphism $h\colon\m{Tm}\to\WarpA$ and $n\in\som$ to the existence of a $\De$-diagram $\de$ that maps $\ka$ to $n$ and $\s{t}{\ka}$ to $\sem{t}_h(n)$ for all $t\in T$, where $\De$ is the saturation of the sample set $\{\s{t}{\ka}\mid t\in T\}$.

\begin{prop}\label[prop]{proposition:valuation-to-diagram}
Let~$T$ be a set of basic terms, $\ka$ a time variable, and $\De$ the saturation of the sample set $\{\s{t}{\ka}\mid t\in T\}$. Then for any homomorphism $h\colon\m{Tm}\to\WarpA$ and $n\in\som$, there exists a~$\De$-diagram~$\de$ such that~$\de(\ka)=n$ and~$\de(\s{t}{\ka})=\sem{t}_h(n)$ for all $t\in T$.
\end{prop} 

\begin{proof}
 We define the map $\de\colon \De \to \som$ recursively by
 \begin{align*}
  \de(\ka) & \defeq n\\
  \de( \s{t}{\al}) & \defeq  \sem{t}_h(\de(\al)) && \text{for each }\s{t}{\al}\in\De\\
  \de(\last(t)) & \defeq \lastw(\sem{t}_h) && \text{for each } \last(t) \in \De\\ 
   \de(\suc(\al)) & \defeq  \aS(\de(\al)) &&  \text{for each } \suc(\al) \in \De.
 \end{align*}
The map $\de$ is well-defined since $\al \in \De$ if, and only if, there exist samples $\al_1,\ldots,\al_n$ such that $\al_1 = \s{t}{\ka}$ for some $t\in T$, $\al_n = \al$, and $\al_j \leadsto \al_{j+1}$ for each $j\in\{1,\ldots,n-1\}$. So $\de$ is a $\fosom$-valuation.  It remains to prove that $\de$ is a diagram, i.e., that $\de$ satisfies \Cref{c:mon}-\Cref{c:inv-last-inf}. Let us assume for convenience without further mention that all samples used are in $\De$, and agree to write $\sem{t}$ for $\sem{t}_h$.

 \begin{itemize}[labelindent=.3in]
 
 \item[\Cref{c:mon}] 
If $\de(\al) \leq \de(\be)$, then, by the definition of $\de$, using the fact that time warps are order-preserving, $\de(\s{t}{\al}) = \sem{t}(\de(\al)) \leq \sem{t}(\de(\be)) = \de(\s{t}{\be})$.
 
 \item[\Cref{c:zero}] 
 If $\de(\al) = 0$, then $\de(\s{t}{\al}) = \sem{t}(\de(\al)) = 0$.
 
 \item[\Cref{c:suc}] 
Follows directly from the definition of $\de$.
 
 \item[\Cref{c:last}] 
Using the definition of $\de$, 
  \[
 \de(\last(t)) = \lastw(\sem{t}) = \min\{ n\in \som \mid \sem{t}(n) = \sem{t}(\om) \}.
 \]
 Hence, clearly, for each $k\in \som$,
 \[
 \lastw(\sem{t}) \leq k \:\iff\: \sem{t}(\lastw(\sem{t})) = \sem{t}(k),
 \]
 and therefore, for all $\s{t}{\al} \in \De$,
 \[
 \de(\last(t)) \leq \de(\al) \:\iff\: \de(\s{t}{\last(t)}) = \de(\s{t}{\al}).
 \]

 \item[\Cref{c:last2}] 
If $\lastw(\sem{t}) = \de(\last(t)) = \om$, then $\sem{t}(n) < \sem{t}(\om)$ for all $n\in\om$, and it follows that $\de(\s{t}{\last(t)}) = \sem{t}(\om) = \bigvee\{\sem{t}(n)\mid n\in\om\}=\om$.

 \item[\Cref{c:id}] 
$\de(\s{\ut}{\al}) =\id(\de(\al))  = \de(\al)$.
 
 \item[\Cref{c:id2}] $\de(\last(\ut)) =  \lastw(\id) = \om$.
 
 \item[\Cref{c:prod}] 
$\de(\s{tu}{\al}) = \sem{tu}(\de(\al)) = \sem{t}(\sem{u}(\de(\al))) = \de(\s{t}{\s{u}{\al}})$.

 \item[\Cref{c:last-prod1}] 
If $\de(\last(tu)) = \om$, then $\lastw(\sem{t}\sem{u}) = \om$ and, using \Cref{lemma:last-properties}, it follows that $\de(\last(t)) = \lastw(\sem{t}) = \om$ and $\de(\last(u)) =\lastw(\sem{u}) = \om$.

\item[\Cref{c:inv-lower}]If $0<\de(\al) < \om$, then \Cref{prop:inv} yields $\de(\s{t}{\s{\invt{t}}{\al}}) = \sem{t}(\inv{\sem{t}}(\de(\al))) < \de(\al)$.

\item[\Cref{c:inv-upper}] If $\inv{\sem{t}}(\de(\al)) = \de(\s{\invt{t}}{\al}) < \om$, then either $\de(\al) = 0 \leq \de(\s{t}{\suc(\s{\invt{t}}{\al})})$ or $\de(\al)>0$ and \Cref{prop:inv} yields $\de(\al) \leq \sem{t}(\inv{\sem{t}}(\de(\al)) + 1) =  \de(\s{t}{\suc(\s{\invt{t}}{\al})})$.

\item[\Cref{c:inv-last-inf}]
It follows directly from \Cref{lemma:last-properties} that $\lastw(\inv{\sem{t}}) = \de(\last(\invt{t})) = \om$ if, and only if, $\de(\last(t)) = \lastw(\sem{t}) = \om$.  \qedhere
\end{itemize}  
\end{proof}

\begin{runexa}\label{construction}
Consider again $\De = \sat{\{ \s{x \invt{x}}{\ka}\}}$. Applied to a homomorphism $h\colon\m{Tm}\to\WarpA$ with $h(x) = \id$ and $n = 1$, the construction in the above proof yields exactly the diagram $\de\colon \De \to \som$ defined in part~\ref{diagram} of the running example.
\end{runexa}

Next, we provide an algorithm that converts any $\De$-diagram $\de$ for a finite saturated sample set $\De$ into an algorithmic description of a homomorphism  $h\colon\m{Tm}\to\WarpA$ satisfying $\sem{t}_h(\de(\al)) = \de(\s{t}{\al})$ for all $t[\al] \in \De$. Indeed, we show that this homomorphism can be described by mapping each $x\in{\rm Var}$ occurring in $\De$ to a {\em regular} time warp, i.e., a time warp that is either eventually constant or eventually linear. The main challenge here is to show that each partial map consisting of ordered pairs $(\de(\al), \de(\s{x}{\al}))$ with $\s{x}{\al} \in \De$ extends to a time warp and that a composition of these time warps given by a basic term $t$ extends the set of ordered pairs $(\de(\al), \de(\s{t}{\al}))$ with $\s{t}{\al} \in \De$.

 For any finite saturated sample set $\De$, $\De$-diagram $\de$, and basic term~$t$, let
  \begin{mathpar}
    \diag{t}{\de}
    \defeq
    \{
    (\de(\al), \de(\s{t}{\al}))
    \mid
    \s{t}{\al} \in \De
    \}.
  \end{mathpar}
We will say that a time warp~$f$  
\begin{itemize}
\item	 \emph{extends} $\diag{t}{\de}$  if $f(i) = j$ for all~$(i, j) \in \diag{t}{\de}$;
\item	 \emph{strongly extends} $\diag{t}{\de}$  if it extends $\diag{t}{\de}$ and also  	
\begin{align*}
\diag{t}{\de}\neq\emptyset \,\text{ and }\, \de(\last(t)) = \om \implies \lastw(f) = \om.
\end{align*}
\end{itemize}
For convenience, let us fix for the following five lemmas --- which collectively describe how to obtain a time warp that strongly extends $\diag{t}{\de}$ based on the structure of $t$ --- a finite saturated sample set $\De$ and $\De$-diagram $\de$.

\begin{lem}\label[lem]{lemma:nonempty}
 There exists an algorithm that produces an algorithmic description, for any variable $x\in{\rm Var}$, of a time warp~$f$ that strongly extends~$\diag{x}{\de}$.
\end{lem}
\begin{proof}
If~$\diag{x}{\de}=\emptyset$, then any time warp strongly extends it, so we may assume that $\diag{x}{\de}\neq\emptyset$.  By~\cref{c:mon}, if $\s{x}{\al},\s{x}{\be} \in \De$ with $\de(\al) \leq \de(\be)$, then also $\de(\s{x}{\al}) \leq \de(\s{x}{\be})$. Hence $\diag{x}{\de}$  is a partial order-preserving map from~$\som$ to $\som$. Moreover, since~$\De$ is saturated, it follows that~$(\de(\last(x)), \de(\s{x}{\last(x)})) \in \diag{x}{\de}$. Now for $\s{x}{\al} \in \De$, if $\de(\al) = 0$, then, by~\cref{c:zero},  $\de(\s{x}{\alpha}) = 0$, and if $\de(\al) \geq \de(\last(x))$, then, by~\cref{c:last}, $\de( \s{x}{\al}) ) = \de(\s{x}{\last(x)})$. Hence
the relation $X \coloneqq \diag{x}{\de} \cup \{ (0, 0), (\om,\de(\s{x}{\last(x)})) \}$ is also a partial order-preserving map. For each $i \in\om$, there exists a unique pair~$(i_1, j_1), (i_2, j_2) \in X$ such that~$i_1 \le i < i_2$ and there is no~$(i_3, j_3) \in X$ with~$i_1 < i_3 < i_2$, and we define
 \[
 f(i) \defeq \min(j_2, j_1 \oplus (i - i_1)),
 \]
 where $n\oplus m := \min\{\om, n+m \}$. Let also $f(\om) \defeq \de(\s{x}{\last(x)}))$.

Clearly, since $X$ is order-preserving, also $f$ is order-preserving. Moreover, it extends~$\diag{x}{\de}$, since~$i = i_1 < \om$ implies~$f(i_1) = \min(j_2, j_1) = j_1$. In particular, $f(0) = 0$. To confirm that $f$ is a time warp, it remains to show  that~$f(\om) = \bigvee\{f(i)\mid i\in\om\}$. If~$\de(\s{x}{\last(x)}) =f(\om) < \om$, then, by \cref{c:last2}, $\de(\last(x)) < \om$ and, since $f$ is order-preserving, $f(i) = f(\om)$ for each~$i \ge \de(\last(x))$ and~$f(\om) = f(\de(\last(x))) = \bigvee\{f(i)\mid i\in\om\}$.  If~$f(\om) = \om$, then for each~$j\in\om$, there exists an~$i\in\om$ such that~$f(i) > j$, and hence~$\bigvee\{f(i)\mid i\in \om\} = \om = f(\om)$. 

Finally, suppose that~$\de(\last(x)) = \om$. Then \cref{c:last2} yields $(\om,\om) \in \diag{x}{\de}$ and, by \cref{c:last}, for any $(i,j) \in \diag{x}{\de}$, if $i\in\om$, then also $j\in\om$. Hence, $\lastw(f) = \om$, by the definition of $f$. So $f$ strongly extends~$\diag{x}{\de}$.
\end{proof}

\begin{runexa}
Consider again $\De =  \sat{\{\s{x \invt{x}}{\ka}\}}$ and the $\De$-diagram $\de\colon \De \to \som$ from parts~\ref{diagram} and~\ref{construction} of the running example. The construction in the above proof yields the strong extension of $\diag{x}{\de}$ depicted in \Cref{fig:2}, where the added relations are indicated with dashed arrows. Indeed, the obtained extension is the time warp $\id$, which we already know to be a suitable extension.
\begin{figure}
\centering
\begin{tikzpicture}[place/.style={circle,draw=black,fill=black, minimum size = 3pt, inner sep = 0pt, outer sep = 2pt}]
  \node[place]        (nw) at (0,2.4) {};
  \node[place] (n3) at (0,1.4) {};
  \node[place] (n2) at (0,0.7)  {};
  \node[place] (n1) at (0,0) {};
  \node[place] (n0) at (0,-0.7)    {};
  \node[place] (mw) at (2,2.4)  {};
  \node[place] (m3) at (2,1.4) {};
  \node[place] (m2) at (2,0.7)  {};
  \node[place] (m1) at (2,0) {};
  \node[place] (m0) at (2,-0.7)    {};
  \node[left] at (nw) {$\om$};
  \node[left] at (n3) {$3$};
  \node[left] at (n2) {$2$};
  \node[left] at (n1) {$1$};
  \node[left] at (n0) {$0$};
  \node[right] at (mw) {$\om$};
  \node[right] at (m3) {$3$};
  \node[right] at (m2) {$2$};
  \node[right] at (m1) {$1$};
  \node[right] at (m0) {$0$};
  \draw[loosely dotted, line width=0.62pt] (n3) -- (nw);
  \draw[loosely dotted, line width=0.62pt] (m3) -- (mw);
  \draw[->,dashed, thick] (n3) -- (m3);
  \draw[->,dashed, thick] (n2) -- (m2);
  \draw[->, thick] (n1) -- (m1);
  \draw[->, thick] (n0) -- (m0);
  \draw[->, thick] (nw) -- (mw);
\end{tikzpicture}
\caption{Visualization of the extension of $\diag{x}{\de}$}
\label{fig:2}
\end{figure}
\end{runexa}

\begin{lem} \label[lem]{lemma:fundamental-product}
Let $t_1$ and $t_2$ be basic terms. If~$f_1$ and $f_2$ strongly extend~$\diag{t_1}{\de}$ and~$\diag{t_2}{\de}$, respectively, then $f_1 f_2$ strongly extends~$\diag{t_1 t_2}{\de}$.
\end{lem}
\begin{proof}
Suppose that $f_1$ and $f_2$ strongly extend~$\diag{t_1}{\de}$ and~$\diag{t_2}{\de}$, respectively. Then for all~$\s{t_1 t_2}{\al} \in \De$, 
  \begin{align*}
    f_1 f_2 (\de(\al))
    & =
    f_1 (f_2(\de(\al)))
    & \mbox{(by definition)}
    \\
    & =
    f_1(\de(\s{t_2}{\al}))
    & \mbox{(since $f_2$ extends~$\diag{t_2}{\de}$)}
    \\
    & =
    \de(\s{t_1}{\s{t_2}{\al}})
    & \mbox{(since $f_1$ extends~$\diag{t_1}{\de}$)}
    \\
    & =
    \de(\s{t_1 t_2}{\al})
    & \mbox{(by \cref{c:prod})}.
  \end{align*}
 So~$f_1 f_2$ extends~$\diag{t_1 t_2}{\de}$, and it remains to show that the extension is strong. We can assume that $\diag{t_1 t_2}{\de}$ is non-empty, since otherwise there is nothing to prove. Suppose that~$\de(\last(t_1 t_2)) = \om$. Then $\de(\last(t_1)) = \de(\last(t_2)) = \om$, by~\cref{c:last-prod1}, and, since $f_1$ and $f_2$ strongly extend $\diag{t_1}{\de}$ and $\diag{t_2}{\de}$, respectively, also $\lastw(f_1) = \lastw(f_2) = \om$. Hence $\lastw(f_1 f_2) = \om$, by~\Cref{lemma:last-properties}.
\end{proof}

\begin{lem} \label[lem]{lemma:fundamental-star} 
Let $t$ be a basic term. If~$f$ strongly extends~$\diag{t}{\de}$, then~$\inv{f}$ strongly extends~$\diag{\invt{t}}{\de}$.
\end{lem}
\begin{proof}
Let $\s{\invt{t}}{\al} \in \De$. Note that, by \cref{c:zero}, if $\de(\al) = 0$, then $\de(\s{\invt{t}}{\al}) = 0=\inv{f}(0)$. Hence we can assume  that $\de(\al)> 0$.  Suppose first that $\de(\s{\invt{t}}{\al}) < \om $.  Then, by \cref{c:inv-upper},  $ \de(\al) \leq \de(\s{t}{\suc(\s{\invt{t}}{\al})})$.  Since $f$ extends~$\diag{t}{\de}$, it follows that $ \de(\al) \leq \de(\s{t}{\suc(\s{\invt{t}}{\al})}) = f(\de(\s{\suc(\invt{t}}{\al})))$, and, by~\cref{c:suc}, also $f(\de(\s{\suc(\invt{t}}{\al})))= f(\de(\s{\invt{t}}{\al})+1) $, i.e., $\de(\al) \leq f(\de(\s{\invt{t}}{\al})+1)$. We have two cases:
\begin{enumerate}
\item $\de(\al) < \om$. Then, by \Cref{c:inv-lower} and since $f$ extends $\diag{t}{\de}$, it follows that 
$f( \de({\s{\invt{t}}{\al}})) = \de(\s{t}{\s{\invt{t}}{\al}}) < \de(\al)$. So $\de(\s{\invt{t}}{\al}) =\inv{f}(\de(\al))$, by \Cref{prop:inv}.
\item $\de(\al) = \om$. Then, since $\de(\s{\invt{t}}{\al}) < \om$, by \cref{c:last} and \cref{c:last2},  it follows that  $\last(\invt{t}) < \om$ and $\de(\s{\invt{t}}{\last(\invt{t})}) = \de(\s{\invt{t}}{\al}) $. Hence, as in the first case,
\(
f(\de({\s{\invt{t}}{\last(\invt{t})}})) < \de(\last(\invt{t}))
\).
So $f( \de({\s{\invt{t}}{\al}})) < \om = \de(\al) \leq  f(\de(\s{\invt{t}}{\al})+1)$, and $\de(\s{\invt{t}}{\al}) =\inv{f}(\de(\al))$, by \Cref{prop:inv}.
\end{enumerate}
Otherwise $\de(\s{\invt{t}}{\al}) = \om$ and we again have two cases:  
\begin{enumerate}
 \item $\de(\al)<\om$. Then, since $f( \de({\s{\invt{t}}{\al}})) = \de(\s{t}{\s{\invt{t}}{\al}}) < \de(\al)$ by  \Cref{c:inv-lower}, it follows that $f(\om) =  f(\de(\s{\invt{t}}{\al})) < \de(\al)$. Hence $\inv{f}(\de(\al)) = \om$, by  \Cref{prop:inv}.
 \item $\de(\al) = \om$. Then there are two cases. If $\de(\last(\invt{t})) <\om$, then, using the previous cases, $\inv{f}(\de(\last(\invt{t}))) = \om$, and hence  $\inv{f}(\om) = \om$. Otherwise $\de(\last(\invt{t})) = \om$. Then $\de(\last(t)) = \om$, by \cref{c:inv-last-inf}, and, since $f$ strongly extends~$\diag{t}{\de}$, also $\lastw(f) = \om$. But then $\inv{f}(\om) = \om$, by \Cref{prop:inv}.
  \end{enumerate}
  It remains to show that the extension of $\diag{\invt{t}}{\de}$ to $\inv{f}$ is strong. We can assume that $\diag{\invt{t}}{\de}$ is non-empty. Suppose that $\de(\last(\invt{t}))  =\om$. Then $\de(\last(t)) = \om$, by \cref{c:inv-last-inf}, and, since $f$ strongly extends~$\diag{t}{\de}$, also $\lastw(f) = \om$.  Hence $\lastw(\inv{f}) = \om$, by \Cref{lemma:last-properties}.
\end{proof}

\begin{lem} \label[lem]{lemma:fundamental-id}
The time warp~$\id$ strongly extends~$\diag{\ut}{\de}$.
\end{lem}

\begin{proof}
  The extension property follows from~\cref{c:id}; the fact that it is strong follows from the fact that~$\lastw(\id) = \om$. 
\end{proof}

\begin{lem} \label[lem]{lemma:fundamental}
For every basic term~$t$ and homomorphism $h\colon\m{Tm}\to\WarpA$, if~$h(x)$ strongly extends~$\diag{x}{\de}$ for each $x\in{\rm Var}$ occurring in $t$, then~$\sem{t}_h$ strongly extends~$\diag{t}{\de}$.
\end{lem}
\begin{proof}
 By induction on~$t$. The case~$t = x$ is immediate and the other cases follow from Lemmas~\ref{lemma:fundamental-product}--\ref{lemma:fundamental-id}, and the induction hypothesis. 
\end{proof}

The next proposition is then a direct consequence of~Lemmas~\ref{lemma:nonempty} and~\ref{lemma:fundamental}.

\begin{prop}\label[prop]{proposition:diagram-to-valuation}
There exists an algorithm that produces for any finite saturated sample set $\De$ and $\De$-diagram $\de$, an algorithmic description of a homomorphism $h\colon\m{Tm}\to\WarpA$ satisfying $\sem{t}_h(\de(\al)) = \de(\s{t}{\al})$ for all $t[\al] \in \De$.
\end{prop}

We have now assembled all the necessary ingredients to relate the validity of an equation in the time warp algebra to the existence of a corresponding diagram.

\begin{prop}
  \label[prop]{proposition:diagram}
Let $t_1,\ldots,t_n$ be basic terms, $\ka$ a time variable, and $\De$ the saturation of the sample set  $\{\s{t_1}{\ka},\ldots,\s{t_n}{\ka}\}$. Then $\WarpA\not\models\id \leq t_1 \jn \cdots \jn t_n$ if, and only if, there exists a $\De$-diagram $\de$ such that $\de(\ka) > \de(\s{t_i}{\ka})$ for each $i\in \{1,\ldots,n \}$. 
\end{prop}
\begin{proof}
Suppose first that $\WarpA\not\models\id \leq t_1 \jn \cdots \jn t_n$. Then there exist a homomorphism $h\colon\m{Tm}\to\WarpA$ and an $m\in \som$ such that $m = \id(m) > \sem{t_i}_h(m)$ for each $i\in \{1,\ldots,n \}$. Hence, by \Cref{proposition:valuation-to-diagram}, there exists a $\De$-diagram $\de$ such that  $\de(\ka) = m > \sem{t_i}_h(m) = \de(\s{t_i}{\ka})$ for each $i\in \{1,\ldots,n \}$. 

For the converse, suppose that there exists a $\De$-diagram $\de$ such that $\de(\ka) > \de(\s{t_i}{\ka})$ for each $i\in \{1,\ldots,n \}$.  Then, by \Cref{proposition:diagram-to-valuation}, there exists a homomorphism $h\colon\m{Tm}\to\WarpA$ such that $\sem{t_i}_h(\de(\ka)) = \de(\s{t_i}{\ka})$ for each $i\in \{ 1,\ldots,n\}$. So $\id(\de(\ka) ) = \de(\ka) > \sem{t_i}_h(\de(\ka))$ for each $i\in \{ 1,\ldots,n\}$. Hence  $\WarpA\not\models\id \leq t_1 \jn \cdots \jn t_n$. 
\end{proof}

\begin{runexa}
The $\De$-diagram $\de\colon \De \to \som$ for $\De =  \sat{\{ \s{x \invt{x}}{\ka} \}}$ from parts~\ref{diagram} and~\ref{construction} of the running example witnesses the failure of the equation $\ut \leq x \invt{x}$ in $\WarpA$, since $\de(\ka) = 1 > 0 = \de( \s{x \invt{x}}{\ka})$.
\end{runexa}

Observe now that the time warps constructed in the proof of \Cref{lemma:nonempty} are regular. Hence, since, by \Cref{prop:regular-subalgebra}, the set of regular time warps forms a subalgebra $\m R$ of $\WarpA$, we obtain the following correspondence:

\begin{thm}\label{thm:regular}
For any $s,t\in{\rm Tm}$, $\WarpA \models s\eq t$ if, and only if, $\m{R}\models s\eq t$.
\end{thm}

Finally, let $t_1,\ldots,t_n$ be basic terms, $\ka$ a time variable, and $\De$ the saturation of the sample set  $\{\s{t_1}{\ka},\ldots,\s{t_n}{\ka}\}$. Define also $\mathsf{fail} := \{  \s{t_1}{\ka} \ale\ka,\dots,\s{t_n}{\ka} \ale \ka \}$ and let $\psi_\De$ be the existential closure of  the $\tau$-formula $ \bigandd (\Si_\De \cup \mathsf{fail})$.
   Then, by \Cref{proposition:diagram}, it follows that $\WarpA\not\models\id \leq t_1 \jn \cdots \jn t_n$ if, and only if, there exists a $\fosom$-valuation $\de$ such that $\fosom,\de \models \Si_\De \cup \mathsf{fail}$ if, and only if,  $\fosom \models \psi_\De$.  Since the first-order theory of $\langle \som, \aleq \rangle$ is decidable (see, e.g., \cite{LL66}) and the operations $0$, $\om$, and $\aS$ are definable in $\langle \som, \aleq \rangle$, we obtain the following result:
   
\begin{thm}\label{thm:decidable}
The equational theory of $\WarpA$ is decidable.
\end{thm}


\section{Implementation}\label{sec:implementation}

We have implemented the decision procedure described in Section~\ref{sec:diagrams} as a software tool. The tool is written in the OCaml programming language~\cite{OCaml} and makes use of the Z3 theorem prover~\cite{Z3}. It reads an equation between terms in the language $\{\mt,\jn, \cdot, \ld, \rd,\mathop{\invt{}}, \ut \}$ from the command line and translates it into a Satisfiability Modulo Theory~(SMT) query. This query expresses the existence of corresponding diagrams for the equation, constructed as in Section~\ref{sec:diagrams}, and is unsatisfiable if, and only if, the equation is satisfied by the time warp algebra. If the query is satisfiable, the tool translates the model provided by Z3 into a homomorphism from $\m{Tm}$ to $\WarpA$ falsifying the initial problem.

\subsection*{Outline}

We summarize the main steps of the decision procedure below, mentioning details of the implementation in passing.

\begin{enumerate}
\item
  The input problem is normalized into an inequation of the form~$\ut\le t$.
  This is always possible since:
  \begin{itemize}
  \item
    any inequation~$u \le v$ can be rewritten into the
    inequation~$\ut\le u \ld v$;

  \item
    an equation~$u \eq v$ can be rewritten into the inequation
    $\ut\le (u\ld v)\mt (v\ld u)$.

  \end{itemize}

\item
  The inequation $\ut\le t$ is rewritten, as per~\Cref{lem:normal}, into an inequation of
  the form
  \[
 \ut\le\bigmt_{i=1}^m \bigjn_{j=1}^{n_i} t_{i,j},
  \]
  where each~$t_{i,j}$ is a basic term, by first iteratively replacing subterms of the form $r\ld s$ and $s\rd t$ with $(s'r)'$ and $(rs')'$, respectively, then  iteratively distributing joins over meets and the monoid multiplication over meets and joins, pushing the involution inwards using the De Morgan laws.   Each inequation $\ut\le\bigjn_{j=1}^{n_i} t_{i,j}$ can then be checked independently for $i\in\{1,\dots,m\}$.

\item
  Optionally, each basic term~$t_{i,j}$ can be rewritten using simple
  algebraic rules.
  In effect, this step computes normal forms for the following convergent
  rewrite system:
  \begin{mathpar}
    t \cdot \ut \to t

    \ut \cdot t \to t

    t'' \to t.
  \end{mathpar}
  This simplification step is optional but easy to implement and can
  drastically reduce the cost of the rest of the decision procedure.

\item
  The tool computes for each~$i$ in~$\{1,\dots,m\}$,  the saturated sample set
  \[
    \De_i
    :=
    \sat{
      \{
      \s{t_{i, 1}}{\ka_i}, \dotsc, \s{t_{i, n_i}}{\ka_i}
      \}
    },
  \]
 where~$\ka_1, \dotsc, \ka_n$ are distinct sample variables.
  Samples are represented as first-order terms with maximal
  sharing~\cite{FilliatreConchon-2006}, so that sample equations can be tested in
  constant time.
  This accelerates the membership tests used to check whether saturation is
  complete.

\item
  For each~$i$ in~$\{1,\dots,m\}$, the set~$\De_i$  is traversed to generate an SMT query built from the  
  quantifier-free first-order formulas described in~\Cref{sec:diagrams}.
  These first-order formulas are expressed over the theory known as~\emph{Integer Difference
    Logic} in the SMT literature~\cite{SMTLIB}, which corresponds to the
  first-order theory of the
  structure~$\integers=\langle \mathbb{Z}, \Zle, \Zz,
  \ZaS \rangle$, where $\Zle$ is the natural order on $\mathbb{Z}$,~$\Zz =0$, and $\ZaS(m) = m+1$, for each $m\in\mathbb{Z}$.
  We encode~$\fosom$ in this structure by representing~$\om^{\fosom}$
  as~$\Zz$ and~$n\in\om$ as~$\ZaS n$.
  More precisely, writing~$\encode{-}$ for the encoding map, we let
  \begin{align*}
    \encode{0^{\fosom}}
    & :=
    \aS^{\integers} \Zz
    \\
    \encode{\om^{\fosom}}
    & :=
    \Zz
    \\
    \encode{\xi \eq \aS^{\fosom} t}
    & :=
    (\encode{t} \eq \Zz \andd \xi \eq \Zz)
    \orr
    (\ZaS \Zz \Zle \encode{t}
    \andd
    \xi \eq \ZaS \encode{t})
    \\
    \encode{t \preceq^{\fosom} u}
    & :=
    (\encode{u} \eq \Zz)
    \orr
    (\ZaS \Zz \Zle \encode{t}
    \andd
    \encode{t} \Zle \encode{u}),
  \end{align*}
  where~$\xi$ denotes an arbitrary first-order variable.
  Since the operation~$\aS^{\fosom}$ is encoded as a formula rather than a
  term in the target theory, we have assumed without loss of generality that it
  only appears at the root of an equation.
  Additionally, for each first-order variable~$\xi$, the tool generates the
  formula~$\Zz \Zle \xi$.

\item
  Finally, the~$i$th query is passed to the SMT solver.
  It is satisfiable if, and only if, $\WarpA\not\models\ut\le\bigjn_{j=1}^{n_i} t_{i,j}$. 
  In this case, the tool uses the diagram returned by the SMT solver to construct a homomorphism $h\colon\m{Tm}\to\WarpA$ and element $k\in\som$ satisfying $k>\sem{t_{i,j}}_h(k)$, for each $j\in\{1,\dots,n_i\}$,  following the approach described in the proof of~\Cref{lemma:nonempty}.

\end{enumerate}

\subsection*{Examples}

The running example~$\ut\le x x'$ considered in \Cref{sec:diagrams} can be handled 
by our implementation. 
The generated SMT query contains~139 clauses.
Z3 checks its satisfiability and constructs the diagram
\begin{align*}
  0 &
  = \de(\s{x x'}{\ka})
  = \de(\s{x}{\s{x'}{\ka}})
  = \de(\s{x x'}{\last(x x')})
  = \de(\s{x}{\s{x'}{\last(x x')}})
  = \de(\s{x'}{\last(x x')})
  \\ &
  = \de(\s{x}{\s{x'}{\last(x')}})
  = \de(\s{x}{\suc(\s{x'}{\last(x x')})})
  = \de(\last(x x'))
  \\
  1 &
  = \de(\suc(\s{x'}{\last(x x')}))
  = \de(\last(x'))
  \\
  2 &
  = \de(\s{x}{\suc(\s{x'}{\ka})})
  = \de(\s{x'}{\ka})
  = \de(\s{x}{\suc(\s{x'}{\last(x')})})
  = \de(\s{x'}{\last(x')})
  = \de(\ka)
  \\
  3 &
  = \de(\suc(\s{x'}{\last(x')}))
  = \de(\suc(\s{x'}{\ka}))
  = \de(\s{x}{\last(x)})
  \\
  4 &
  = \de(\last(x)),
\end{align*}
which expresses that any time warp~$f$ satisfying~$f(n) = 1$ for~$n <
3$,~$f(3) = 2$, and~$f(4) = 3$, satisfies~$f(\inv{f}(1)) < 1$.

As a further example, consider the inequation~$\ut\le x$, which is small enough to make each step of the decision
procedure explicit. The corresponding saturated sample set is 
\[
  \De
  :=
  \{ \ka, \s{x}{\ka}, \last(x), \s{x}{\last(x)} \}.
\]
The tool traverses~$\De$ to generate a first-order formula expressing
the existence of a diagram.
Expressed in mathematical notation, this formula is
\begin{align*}
  \exists
  \FOv{\ka}, \FOv{\s{x}{\ka}}, \FOv{\last(x)},
  \FOv{\s{x}{\last(x)}}, &
  \bigandd_{\alpha \in \De} \left(\Zz \Zle \alpha\right)
  \\ & \andd
  (\FOv{\ka} \Ole \FOv{\last(x)}
  \Rightarrow
  \FOv{\s{x}{\ka}} \Ole \FOv{\s{x}{\last(x)}})
  \\ & \andd
  (\FOv{\last(x)} \Ole \FOv{\ka}
  \Rightarrow
  \FOv{\s{x}{\last(x)}} \Ole \FOv{\s{x}{\ka}})
  \\ &
  \andd
  \bigandd_{\alpha \in \{ \ka, \last(x) \}}
  \left(
    \FOv{\alpha} \eq \Oz
    \Rightarrow
    \FOv{\s{x}{\alpha}} \eq \Oz
  \right)
  \\ &
  \andd
  \bigandd_{\alpha \in \{ \ka, \last(x) \}}
  \left(
    \FOv{\last(x)} \Ole \FOv{\alpha}
    \Rightarrow
    \FOv{\s{x}{\alpha}} \eq \FOv{\s{x}{\last(x)}}
  \right)
  \\ &
  \andd
  \FOv{\last(x)} \eq \Oom
  \Rightarrow
  \FOv{\s{x}{\last(x)}} \eq \Oom
  \\ &
  \andd \FOv{\s{x}{\ka}} \Zlt \FOv{\ka},
\end{align*}
where the symbols~$\Ole$,~$\Oz$,~$\Oom$, and~$\OaS$ stand for their actual encodings
in terms of~$\Zle$,~$\Zz$, and~$\ZaS$. 
\Cref{fig:example-smt} displays the same formula as a query in concrete SMT-LIB
syntax.
The final step is to check its validity using Z3.
The~SMT solver builds the diagram
\begin{align*}
  0 &
  = \de(\s{x}{\last(x)})
  = \de(\s{x}{\ka})
  \\
  1 & =
  \de(\last(x))
  \\
  2 & =
  \de(\ka),
\end{align*}
expressing in particular that mapping $x$ to the constant time warp~$\bot$ falsifies~$\ut\le x$ in $\WarpA$. 

\begin{figure}
  \begin{verbatim}
(declare-fun !0 () Int)
(declare-fun !1 () Int)
(declare-fun !2 () Int)
(declare-fun !3 () Int)
(assert (<= 0 !0))
(assert (<= 0 !1))
(assert (<= 0 !2))
(assert (<= 0 !3))
(assert (let ((a!1 (or (= !2 0) (and (not (= !3 0)) (<= !3 !2))))
      (a!2 (or (= !0 0) (and (not (= !1 0)) (<= !1 !0)))))
  (=> a!1 a!2)))
(assert (let ((a!1 (or (= !3 0) (and (not (= !2 0)) (<= !2 !3))))
      (a!2 (or (= !1 0) (and (not (= !0 0)) (<= !0 !1)))))
  (=> a!1 a!2)))
(assert (=> (= !3 1) (= !1 1)))
(assert (let ((a!1 (or (= !3 0) (and (not (= !2 0)) (<= !2 !3)))))
  (and (=> a!1 (= !0 !1)) (=> (= !0 !1) a!1))))
(assert (=> (= !2 1) (= !0 1)))
(assert (let ((a!1 (or (= !2 0) (and (not (= !2 0)) (<= !2 !2)))))
  (and (=> a!1 (= !0 !0)) (=> (= !0 !0) a!1))))
(assert (=> (= !2 0) (= !0 0)))
(assert (let ((a!1 (or (= !3 0) (and (not (= !1 0)) (<= !1 !3)))))
  (and a!1 (not (= !1 !3)))))
\end{verbatim}
  \caption{SMT query whose models correspond to diagrams for~$\ut\le x$}
  \label{fig:example-smt}
\end{figure}


\subsection*{Computational complexity}
The worst-case computational complexity of the decision procedure is exponential
in the size of the initial equation.
This complexity breaks down as follows.
Steps~2 and~4 have a running time that is exponential in the size of their input data.
In particular, the size of the saturated sample set for a given basic
term~$t_{i,j}$ is exponential in the size of~$t_{i,j}$ in general.
The size of the SMT query generated during step~5 is polynomial in the number of
samples in the input sample set.
Finally, the SMT solver has a running time that is exponential in the size of the SMT
query in the worst case, as the conjunctive fragment of difference logic is
decidable in polynomial time~\cite{Dill-1989}.


\subsection*{Software complexity}
The entirety of the implementation, not counting Z3 and its OCaml interface,
takes slightly less than 1900 lines of code.
The core procedure and data structures represent around 550 lines of code.
The code is available online as libre
software.\footnote{\url{https://github.com/adrieng/twmc}}


\section{Concluding Remarks}\label{sec:concluding}

The main contribution of this paper is a proof that the equational theory of the time warp algebra $\WarpA$ is decidable, supported by an implementation of a corresponding decision procedure. There remain, however, many open problems regarding computational and structural properties of $\WarpA$ and related algebras of join-preserving maps on complete chains.  Below, we briefly discuss some of these problems and potential avenues for further research.


\subsection*{An axiomatic description}

Although we have provided an algorithm to decide equations in the time warp algebra, an axiomatic description of its equational theory is still lacking. The algebra generates a variety of fully distributive involutive residuated lattices satisfying the equations $\ut'\le\ut$ and $\ut\eq\ut'\pd(\ut'\ld\ut)$, but not much more is known. An axiomatization could also provide the basis for developing a sequent calculus for reasoning in $\WarpA$ along the lines of Yetter's cyclic linear logic without exponentials~\cite{Yet90}.


\subsection*{Computational complexity}

As explained in Section~\ref{sec:implementation}, the decision procedure described in this paper provides an exponential upper bound for the computational complexity of deciding equations in the time warp algebra. A lower bound for this problem follows from the fact that the involution-free reduct of $\WarpA$ generates the variety of distributive $\ell$-monoids, which has a co-NP-complete equational theory~\cite{CGMS21}. Tight complexity bounds have yet to be determined, however.


\subsection*{Applications to graded modalities.}

The implementation described in Section~\ref{sec:implementation} can be readily
integrated in a type checker for a programming language with a graded modality
whose gradings are first-order terms over $\WarpA$, generalizing
previous work~\cite{Guatto-2018}.
In broad strokes, such a type checker reduces each subtyping problem arising
during type checking to a finite set of inequalities in the time warp
algebra.
The simplest interesting case is that of checking whether the type~$\Box_t A$ is
a subtype of~$\Box_u A$, which reduces to checking~$\m{W} \models u \le t$.
The other cases in reducing type-checking problems to
universal-algebraic problems depend on the exact language of types under
consideration, including its subtyping relationship, which is beyond the scope
of the current paper.
Let us mention, however, that certain programming language features, such as the
modeling of fixed computation rates, may require extending the language of
gradings, and hence also the decision procedure.
We discuss a simple example of such an extension in the next paragraph.


\subsection*{Extending the language}

For certain applications, it may be profitable to extend the language of time warps with further operations. To illustrate, let us describe how `definable' time warps can be added as constants to the language, while preserving decidability of the resulting equational theory. Let $X$ be a countably infinite set of fresh variables for building first-order formulas. We call a time warp $f$ \emph{definable} if there exist first-order formulas $\phi_f(x,y)$ and $\psi_f(z)$ in the signature $\tau = \{\aleq, \aS, 0, \om \}$ with free variables $x,y$, and $z$, respectively such that for every $\fosom$-valuation $v\colon X \to \som$,
\[
\fosom, v \models \phi_f(x,y)\:\iff\:f(v(x)) = v(y)\quad\text{ and }\quad\fosom, v \models \psi_f(z) \:\iff\: v(z) = \lastw(f).
\]
For example, the time warp $\top$ satisfying $\top(0) = 0$ and $\top(n) = \om$, for each $n\in \som{\setminus}\{0\}$, is definable by  $\phi_\top(x,y) \coloneqq (x \eq 0 \To y \eq 0) \andd (0\ale x \To y \eq \om)$ and $\psi_\top(z) \coloneqq z \eq \aS(0)$. Similarly, the time warp $\bot$ satisfying $\bot(n) = 0$, for each $n\in \som$, is definable by $\phi_\bot(x,y) \coloneqq y \eq 0$ and $\psi_\bot(z) \coloneqq z\eq 0$.

Let $\Dwarp$ denote the set of definable time warps. Given any $F = \{f_1,\dots, f_n\}\subseteq \Dwarp$, let $\WarpA[F]$ denote the algebra  $\langle \res(\som), \mt , \jn, \circ, \inv{}, \id, f_1,\dots, f_1 \rangle$. We call a term in the signature  $\{\cdot,\mathop{'}, \ut, f_1, \dots, f_n \}$  \emph{$F$-basic} and define \emph{$F$-samples} analogously to samples by  allowing $F$-basic terms in the construction. The notion of saturation is extended to $F$-sample sets, and following the proof of  \Cref{l:fin-sat} shows that also the saturation of a finite $F$-sample set is finite. The notion of a diagram is then extended to saturated $F$-sample sets $\De$ by defining the sets $\mathsf{struct}_\De$, $\mathsf{log}_\De$,  $\mathsf{inv}_\De$ as before, and also
\[
\mathsf{const}_\De \coloneqq \{\phi_f(\al,\s{f}{\al}) \mid \s{f}{\al} \in \De, f \in F \} \;  \cup \{\psi_f(\last(f)) \mid \last(f) \in \De, f\in F\}.
\]
Let $\Si_\De \coloneqq \mathsf{struct}_\De \cup \mathsf{log}_\De \cup \mathsf{inv}_\De \cup \mathsf{const}_\De$. A \emph{$\De$-diagram} is a $\fosom$-valuation $\de\colon \De \cup X \to \som$ such that $\fosom, \de \models \Si_\De$. It is clear that \Cref{proposition:valuation-to-diagram} extends to saturated $F$-sample sets, using the definition of a definable time warp and the fact that homomorphisms from the term algebra into $\WarpA[F]$ map $f$ to $f$, for each $f\in F$. Similarly, since every $f\in F$ is completely described by the formulas $\phi_f(x,y)$ and $\psi_f(z)$, also \Cref{proposition:diagram-to-valuation}, extends  to saturated $F$-sample sets.  Hence, arguing as in \Cref{sec:diagrams}, the equational theory of $\WarpA[F]$ is decidable.

 
\subsection*{A relational approach}
 
We end this paper by discussing a different, relational point of view on time warps, 
which we believe will play a role in generalizing and extending the methods of this paper.
It is well-known that, if $\m{L}$ is a sufficiently well-behaved lattice, then any completely 
join-preserving function on $\m{L}$ can be uniquely described via a certain binary relation on 
the set $P$ of completely join-prime elements of $\m{L}$. 
The binary relations that occur in this way are exactly the \emph{distributors} on $P$, 
when viewed as a posetal category. This raises the 
question as to whether the decision procedure we give here in the case 
$\m{L} = \om^+$ can be 
generalized to the algebraic structure of distributors on an object of a sufficiently
well-behaved category $P$. 
In order to illustrate the idea, we explain here how the results of the paper can
be alternatively understood in this relational language.

Let us write $P = \{1, 2, \dots \}$ for the set of positive natural numbers, so that $\som$ is isomorphic to the lattice of downward closed subsets of $P$, via the function that sends $x \in \som$ to ${\downarrow} x \cap P$. A time warp $f \colon \som \to \som$ may then be viewed alternatively as a binary relation $R \subseteq P \times P$ that is \emph{monotone} (also called a {\em weakening relation} in, e.g.,~\cite{GJ20}): that is, for any $x, x', y, y' \in P$, if $x' \leq x$, $x R y$, and $y \leq y'$, then $x' R y'$. Indeed, given a time warp $f$, the binary relation $R_f$ defined by
\[ 
R_f := \{(x,y) \in P^2 \mid x \leq f(y)\}
\] 
is clearly monotone, and conversely, given a monotone binary relation $R$ on $P$, the function $f \colon P \to \som$ defined, for $y \in P$, by  
\[ 
f(y) := \bigvee \{ x \in P \mid xRy\} 
\] 
is order-preserving and hence extends uniquely to a time warp by setting $f(0):= 0$ and $f(\om) := \bigvee_{p \in P} f(p)$. These assignments $f\mapsto R_f$ and $R \mapsto f_R$ yield an order-isomorphism between the lattice of time warps and the lattice $\mathcal{M}(P)$ of monotone binary relations on $P$, ordered by inclusion. Let us denote by $\ast$ the associative operation of relational composition on $\mathcal{M}(P)$, that is, $R \ast S := \{(x,z) \in P^2\mid xRy \text{ and } ySz \text{ for some } y \in P\}$. Note that the natural order relation, $\leq_P$, on $P$ is a neutral element in $\mathcal{M}(P)$ for this composition operation $\ast$. Straightforward computations show that, for any $f, g \in W$,
\begin{align*} 
R_{f \circ g} &= R_f \ast R_g \quad\text{ and }\quad  R_{f^\star} &= \{ (x,y) \in P^2 \mid f(x) \leq y - 1 \} = \{(x,y) \in P^2 \mid y \nleq f(x) \},
\end{align*}
where we note that the latter can also be written as the complementary relation of the converse of the relation $R_f$. Hence, the lattice isomorphism between the lattice of time warps and $\mathcal{M}(P)$ extends to an isomorphism between involutive residuated lattices, by equipping the lattice $\mathcal{M}(P)$ with intersection, union, relational composition $\ast$, the unary operation $(-)^\star$ given by $R^\star := P^2{\setminus}(R^{-1})$, and the neutral element $\leq_P$.

From the above considerations, since the isomorphisms are effective, it follows that questions about the theory of the structure $\WarpA$ may be effectively translated into questions about the structure $\mathcal{M}(P)$. Suppose that $t$ is a term in the language of involutive residuated lattices that uses variables $x_1, \dots, x_n$. The above isomorphism allows us to translate the term $t$ into a formula $T(k,k',X_1,\dots,X_n)$ of second-order logic over the structure $\langle P, \leq\rangle$, where the $X_i$ are binary predicate symbols and $k$ and $k'$ are first-order variables.

\bibliographystyle{alphaurl}
\bibliography{bibliography}

\end{document}